\RequirePackage[l2tabu,orthodox]{nag}
\documentclass[final,a4paper,abstracton]{scrartcl}
\pdfoutput=1

\usepackage{todonotes}

\usepackage[utf8]{inputenc}
\usepackage[T1]{fontenc}
\usepackage{lmodern}        
\usepackage[final]{microtype}      
\usepackage{fixltx2e}       

\usepackage{amssymb}
\usepackage{amsmath}
\usepackage{bm}        
\usepackage{mathtools} 

\allowdisplaybreaks[1] 

\usepackage{tikz}
\usepackage{ifthen}
\usetikzlibrary{calc}

\usepackage[inline]{enumitem}
\setlist[itemize]{label=$\circ$}
\setlist[description]{labelindent=\parindent}

\usepackage[%
  backend=biber,
  sorting=nyt,
  citestyle=authoryear-comp, uniquelist=minyear,
  bibstyle=ieee,
  useprefix=true,
  maxbibnames=99,
  maxcitenames=3,
  mincitenames=1,
  firstinits=false,
]{biblatex}
\bibliography{references}

\usepackage[amsmath,hyperref,thmmarks]{ntheorem}

\theoremseparator{.}
\newtheorem{theorem}{Theorem}
\newtheorem{lemma}[theorem]{Lemma}

\newtheorem{corollary}[theorem]{Corollary}

\theoremstyle{empty}
\renewtheorem{theorem*}{Theorem}
\renewtheorem{lemma*}{Lemma}

\theoremstyle{plain}
\theorembodyfont{\upshape}

\newtheorem{definition}[theorem]{Definition}

\theoremstyle{nonumberplain}
\theoremheaderfont{\itshape}
\theorembodyfont{\upshape}
\theoremsymbol{\ensuremath{_\blacksquare}}
\newtheorem{proof}{Proof}

\usepackage[pdfusetitle,pdfborder={0 0 0}]{hyperref}

\DeclarePairedDelimiter\paren{\lparen}{\rparen}

\DeclarePairedDelimiter\set{\{}{\}}
\DeclarePairedDelimiterX\setc[2]{\{}{\}}{\,#1 \;\delimsize\vert\; #2\,}
\DeclarePairedDelimiterX\Prc[2]{\Pr\lparen}{\rparen}{\,#1 \;\delimsize\vert\; #2\,}

\newcommand{\cc}[1]{\ensuremath{\mathsf{#1}}}

\newcommand{\op}[1]{\ensuremath{\operatorname{#1}}}

\renewcommand{\P}{\cc{P}}
\newcommand{\Ppoly}{\cc{P/poly}}
\newcommand{\NP}{\cc{NP}}
\newcommand{\NPpoly}{\cc{NP/poly}}
\newcommand{\coNP}{\cc{coNP}}
\newcommand{\coNPpoly}{\cc{coNP/poly}}
\newcommand{\AM}{\cc{AM}}
\newcommand{\coAM}{\cc{coAM}}

\newcommand{\poly}{\op{poly}}

\newcommand{\Exp}{\mathop{\mathbf{E}}}

\newcommand{\N}{\mathbb{N}}
\newcommand{\Z}{\mathbb{Z}}
\newcommand{\Q}{\mathbb{Q}}
\newcommand{\R}{\mathbb{R}}

\newcommand{\dotcup}{\mathbin{\dot\cup}}

\newcommand{\zo}{\set{0,1}}

\newcommand{\SAT}{\mathrm{SAT}}
\newcommand{\SD}{\mathrm{SD}^{\geq \Delta}_{\leq \delta}}

\newcommand{\sdist}{d}
\newcommand{\KLdiv}{D_{\mathrm{KL}}}
\newcommand{\distD}{\mathcal D}
\newcommand{\Un}{\mathcal U}
\DeclareMathOperator{\supp}{supp}

\newcommand{\domset}{\mathcal D}

\renewcommand{\H}[2]{H\left( #1 \;\middle\vert\; #2 \right)}

\newcommand{\GITDATE}{September 23, 2014}

\begin{document}

\title{\LARGE AND-compression of NP-complete problems:
  Streamlined proof and minor observations}
\author{Holger Dell%
  \texorpdfstring{%
    \\
    \begin{minipage}{.9\textwidth}
      \begin{center}
        \small
        \mbox{}\\
        \textit{Saarland University}\\
        \textit{Cluster of Excellence, MMCI}%
    \thanks{%
      work done as a postdoc at
      LIAFA, Universit\'e Paris Diderot
    }%
      \end{center}
    \end{minipage}
  }{}%
}
\date{\small\GITDATE}

\maketitle

\begin{abstract}
  \textcite{Drucker} proved the following result: Unless the unlikely
complexity-theoretic collapse $\coNP\subseteq\NPpoly$ occurs, there
is no \emph{AND-compression} for SAT.
The result has implications for the compressibility and
kernelizability of a whole range of NP-complete parameterized
problems.
We present a streamlined proof of Drucker's theorem.

An AND-compression is a deterministic polynomial-time algorithm that
maps a set of SAT-instances $x_1,\dots,x_t$ to a
single SAT-instance $y$ of size $\poly(\max_i|x_i|)$
such that $y$ is satisfiable if and only if all $x_i$ are satisfiable.
The ``AND'' in the name stems from the fact that the predicate ``$y$
is satisfiable'' can be written as the AND of all predicates ``$x_i$
is satisfiable''.
Drucker's theorem complements the result by \textcite{BDFH} and
\textcite{FortnowSanthanam}, who proved the analogous statement for
OR-compressions, and Drucker's proof not only subsumes their result but
also extends it to \emph{randomized} compression algorithms that are
allowed to have a certain probability of failure.

\textcite{Drucker} presented two proofs: The first uses information theory and
the minimax theorem from game theory, and the second is an elementary, iterative
proof that is not as general.
In our proof, we realize the iterative structure as a generalization of the
arguments of \textcite{Ko} for $\P$-selective sets, which use the fact that
tournaments have dominating sets of logarithmic size.
We generalize this fact to hypergraph tournaments.
Our proof achieves the full generality of Drucker's theorem, avoids the minimax
theorem, and restricts the use of information theory to a single, intuitive
lemma about the average noise sensitivity of compressive maps.
To prove this lemma, we use the same information-theoretic inequalities as
Drucker.

\end{abstract}

\section{Introduction}

The influential ``OR-conjecture'' by \textcite{BDFH} asserts that $t$ instances 
$x_1,\dots,x_t$ of $\SAT$ cannot be mapped in polynomial time to an instance $y$ 
of size $\poly(\max_i |x_i|)$ so that $y$ is a yes-instance if and only if at 
least one $x_i$ is a yes-instance.
Conditioned on the OR-conjecture, the ``composition framework'' of 
\textcite{BDFH} has been used to show that many different problems in 
parameterized complexity do not have polynomial kernels.
\textcite{FortnowSanthanam} were able to prove that the OR-conjecture holds 
unless $\coNP\subseteq\NPpoly$, thereby connecting the OR-conjecture with a 
standard hypothesis in complexity theory.

The results of
\textcite{BDFH,FortnowSanthanam}
can be used not only to rule out deterministic kernelization 
algorithms, but also to rule out randomized kernelization algorithms 
with one-sided error, as long as the success probability is bigger 
than zero; this is the same as allowing the kernelization algorithm to 
be a $\coNP$-algorithm.
Left open was the question whether the complexity-theoretic hypothesis 
$\coNP\not\subseteq\NPpoly$ (or some other hypothesis believed by 
complexity theorists) suffices to rule out kernelization algorithms 
that are randomized and have two-sided error.
\textcite{Drucker} resolves this question affirmatively; his results 
can rule out kernelization algorithms that have a constant gap in 
their error probabilities.
This result indicates that randomness does not help to decrease the 
size of kernels significantly.

With the same proof, \textcite{Drucker} resolves a second important question: 
whether the ``AND-conjecture'', which has also been formulated by 
\textcite{BDFH} analogous to the OR-conjecture, can be derived from existing 
complexity-theoretic assumptions.
This is an intriguing question in itself, and it is also relevant for 
parameterized complexity as, for some parameterized problems, we can rule out 
polynomial kernels under the AND-conjecture, but we do not know how to do so 
under the OR-conjecture.
\textcite{Drucker} proves that the AND-conjecture is true if 
$\coNP\not\subseteq\NPpoly$ holds.

The purpose of this paper is to discuss Drucker's theorem and its 
proof.
To this end, we attempt to present a simpler proof of his theorem.
Our proof in \S\ref{sec:proof} gains in simplicity with a small loss in 
generality: the bound that we get is worse than Drucker's bound by a factor of 
two.
Using the slightly more complicated approach in \S\ref{sec:proof2}, it 
is possible to get the same bounds as Drucker.
These differences, however, do not matter for the basic version of the 
main theorem, which we state in~\S\ref{sec mainthm} and further 
discuss in~\S\ref{sec comparison}.
For completeness, we briefly discuss a formulation of the composition 
framework in~\S\ref{sec:nopoly}.

\subsection{Main Theorem: Ruling out OR- and AND-compressions}
\label{sec mainthm}

An AND-compression~$A$ for a language $L\subseteq\zo^*$ is a 
polynomial-time reduction that maps a set $\set{x_1,\dots,x_t}$ to 
some instance $y\doteq A\paren[\big]{\set{x_1,\dots,x_t}}$ of a 
language~$L'\subseteq\zo^*$ such that
$y\in L'$ holds if and only if $x_1\in L$ and $x_2\in L$ and $\dots$ 
and $x_t\in L$.
By De Morgan's law, the same $A$ is an OR-compression for 
$\overline{L}\doteq\zo^*\setminus L$ because $y\in \overline{L'}$
holds if and only if
$x_1\in \overline{L}$
or
$x_2\in \overline{L}$
or $\dots$ or
$x_t\in \overline{L}$.
\textcite{Drucker} proved that an OR-compression for~$L$ implies that 
$L\in\NPpoly\cap\coNPpoly$, which is a complexity consequence that is 
closed under complementation, that is, it is equivalent to 
$\overline{L}\in\NPpoly\cap\coNPpoly$.
For this reason, and as opposed to earlier work 
\parencite{BDFH,FortnowSanthanam,DellVanMelkebeek}, it is without loss of 
generality that we restrict our attention to OR-compressions for the remainder 
of this paper.
We now formally state Drucker's theorem.

\begin{theorem}[Drucker's theorem]\label{thm main simple}
  Let $L,L'\subseteq\zo^*$ be languages, let $e_s,e_c \in[0,1]$ be 
  error probabilities with $e_s+e_c<1$, and let $\epsilon>0$.
  Assume that there exists a randomized polynomial-time algorithm~$A$ 
  that maps any set $x=\set{x_1,\dots,x_t}\subseteq\zo^n$
  for some $n$ and $t$
  to $y=A(x)$ such that:
  \begin{description}[font=$\circ$ \normalfont]
    \item[(Soundness)]
      If all $x_i$'s are no-instances of~$L$, then~$y$ is a 
      no-instance of~$L'$ with probability $\geq 1-e_s$.
    \item[(Completeness)]
      If exactly one $x_i$ is a yes-instance of~$L$, then~$y$ is a 
      yes-instance of~$L'$ with probability~$\geq 1-e_c$.
    \item[(Size bound)]
      The size of~$y$ is bounded by $t^{1-\epsilon} \cdot \poly(n)$.
  \end{description}

  Then $L\in\NPpoly \cap \coNPpoly$.
\end{theorem}
The procedure~$A$ above does not need to be a ``full'' OR-compression, 
which makes the theorem more general.
In particular, $A$ is \emph{relaxed} in two ways:
it only needs to work, or be analyzed, in the case that all input 
instances have the same length; this is useful in hardness of kernelization 
proofs as it allows similar instances to be grouped together.
Furthermore, $A$ only needs to work, or be analyzed, in the case that 
at most one of the input instances is a yes-instance of~$L$;
we believe that this property will be useful in future work on 
hardness of kernelization.

The fact that ``relaxed'' OR-compressions suffice in Theorem~\ref{thm 
  main simple} is implicit in the proof of \textcite{Drucker}, but not 
stated explicitly.
Before Drucker's work, \textcite{FortnowSanthanam} proved the special 
case of Theorem~\ref{thm main simple} in which $e_c=0$, but they only 
obtain the weaker consequence $L\in\coNPpoly$, which prevents their 
result from applying to AND-compressions in a non-trivial way.
Moreover, their proof uses the full completeness requirement and does 
not seem to work for relaxed OR-compressions.

\subsection{Comparison and overview of the proof}
\label{sec comparison}

The simplification of our proof stems from two main sources:
\begin{enumerate*}
\item
  The ``scaffolding'' of our proof, its overall structure, is more modular and 
  more similar to arguments used previously by \textcite{Ko}, 
  \textcite{FortnowSanthanam}, and \textcite{DellVanMelkebeek} for
  compression-type procedures and \textcite{isolation} for isolation 
  procedures.
\item
  While the information-theoretic part of our proof uses the same set 
  of information-theoretic inequalities as Drucker's, the simple 
  version in \S\ref{sec:proof} applies these inequalities to 
  distributions that have a simpler structure.
  Moreover, our calculations have a somewhat more mechanical nature.
\end{enumerate*}

Both Drucker's proof and ours use the relaxed OR-compression~$A$ to design a 
$\Ppoly$-reduction from $L$ to the \emph{statistical distance problem}, which is 
known to be in the intersection of $\NPpoly$ and $\coNPpoly$ by previous work 
(cf.~\textcite{XiaoThesis}).
\textcite{Drucker} uses the minimax theorem and a game-theoretic 
sparsification argument to construct the polynomial advice of the 
reduction.
He also presents an alternative proof~\parencite[Section 3]{Druckerfull} in 
which the advice is constructed without these arguments and also without any 
explicit invocation of information theory;
however, the alternative proof does not achieve the full generality of 
his theorem, and we feel that avoiding information theory entirely 
leads to a less intuitive proof structure.
In contrast, our proof achieves full generality up to a factor of two 
in the simplest proof, it avoids game theoretic arguments, and it 
limits information theory to a single, intuitive lemma about the 
average noise sensitivity of compressive maps.

Using this information-theoretic lemma as a black box, we design the 
$\Ppoly$-reduction in a purely combinatorial way:
We generalize the fact that tournaments have dominating sets of 
logarithmic size to \emph{hypergraph tournaments}; these are complete 
$t$-uniform hypergraphs with the additional property that, for each 
hyperedge, one of its elements gets ``selected''.
In particular, for each set~$e\subseteq\overline{L}$ of $t$ 
no-instances, we select one element of~$e$ based on the fact 
that~$A$'s behavior on~$e$ somehow proves that the selected instance 
is a no-instance of~$L$.
The advice of the reduction is going to be a small dominating set of 
this hypergraph tournament on the set of no-instances of~$L$.
The crux is that we can efficiently test, with the help of the 
statistical distance problem oracle, whether an instance is dominated 
or not.
Since any instance is dominated if and only if it is a no-instance 
of~$L$, this suffices to solve~$L$.

In the information-theoretic lemma, we generalize the notion of average noise 
sensitivity of Boolean functions (which can attain two values) to compressive 
maps (which can attain only relatively few values compared to the input length).
We show that compressive maps have small average noise sensitivity.
Drucker's ``distributional stability'' is a closely related notion, which we 
make implicit use of in our proof.
Using the latter notion as the anchor of the overall reduction, 
however, leads to some additional technicalities in Drucker's proof, 
which we also run into in~\S\ref{sec:proof2} where we obtain the same 
bounds as Drucker's theorem.
In~\S\ref{sec:proof} we instead use the average noise sensitivity as 
the anchor of the reduction, which avoids these technicalities at the 
cost of losing a factor of two in the bounds.

\subsection{Application: The composition framework for ruling out 
  \texorpdfstring{$\bm{O(k^{d-\epsilon})}$}{fixed polynomial} kernels}
\label{sec:nopoly}

We briefly describe a modern variant of the composition framework that 
is sufficient to rule out kernels of size $O(k^{d-\epsilon})$ using 
Theorem~\ref{thm main simple}.
It is almost identical to Lemma~1 of~\textcite{DellMarx,DellVanMelkebeek}
and the notion defined by \textcite[Definition~2.2]{HermelinWu}.
By applying the framework for unbounded~$d$, we can also use it to 
rule out polynomial kernels.
\begin{definition}\label{def:composition}
  Let $L$ be a language, and let $\Pi$ with parameter $k$ be a 
  parameterized problem.
  A \emph{$d$-partite composition of $L$ into $\Pi$} is a 
  polynomial-time algorithm~$A$ that maps any set 
  $x=\set{x_1,\dots,x_t}\subseteq\zo^n$ for some $n$ and $t$ to 
  $y=A(x)$ such that:
  \begin{enumerate}[label=(\arabic*)]
    \item\label{composition soundness}
      If all $x_i$'s are no-instances of $L$, then $y$ is a 
      no-instance of $\Pi$.
    \item\label{composition completeness}
      If exactly one $x_i$ is a yes-instance of $L$, then $y$ is a 
      yes-instance of $\Pi$.
    \item\label{composition size}
      The parameter $k$ of $y$ is bounded by $t^{1/d+o(1)}\cdot
      \poly(n)$.
  \end{enumerate}
\end{definition}
This notion of composition has one crucial advantage over previous 
notions of OR-composition:
The algorithm $A$ does not need to work, or be analyzed, in the case 
that two or more of the $x_i$'s are yes-instances.
\begin{definition}
  Let $\Pi$ be a parameterized problem.
  We call $\Pi$ \emph{$d$-compositional} if there exists an $\NP$-hard 
  or $\coNP$-hard problem~$L$ that has a $d$-partite composition 
  algorithm into $\Pi$.
\end{definition}
The above definition encompasses both AND-compositions and OR-compositions 
because an AND-composition of $L$ into $\Pi$ is the same as an OR-composition of 
$\overline L$ into~$\overline \Pi$.
We have the following corollary of Drucker's theorem.
\begin{corollary}
  If $\coNP\not\subseteq \NPpoly$, then no $d$-compositional problem 
  has kernels of size $O(k^{d-\epsilon})$.
  Moreover, this even holds when the kernelization algorithm is 
  allowed to be a randomized algorithm with at least a constant gap in 
  error probability.
\end{corollary}
\begin{proof}
  Let $L$ be an $\NP$-hard or $\coNP$-hard problem that has a 
  $d$-partite composition~$A'$ into $\Pi$.
  Assume for the sake of contradiction that~$\Pi$ has a kernelization algorithm 
  with soundness error at most~$e_s$ and completeness error at most $e_c$ so 
  that $e_s+e_c$ is bounded by a constant smaller than one.
  The concatenation of~$A'$ with the assumed 
  $O(k^{d-\epsilon'})$-kernelization gives rise to an algorithm $A$ 
  that satisfies the conditions of Theorem~\ref{thm main simple}, for 
  example with $\epsilon=\epsilon'/(2d)$.
  Therefore, we get $L\in (\coNPpoly\cap\NPpoly)$ and thus 
  $\coNP\subseteq\NPpoly$, a contradiction.
\end{proof}
Several variants of the framework provided by this corollary are 
possible:
\begin{enumerate}
  \item
    In order to rule out $\poly(k)$-kernels for a parameterized 
    problem~$\Pi$, we just need to prove that $\Pi$ is 
    $d$-compositional for all $d\in\N$; let's call $\Pi$ 
    \emph{compositional} in this case.
    One way to show that $\Pi$ is compositional is to construct a 
    single \emph{composition} from a hard problem~$L$ into $\Pi$; this 
    is an algorithm as in Definition~\ref{def:composition}, except 
    that we replace \ref{composition size} with the bound $k\leq 
    t^{o(1)} \poly(n)$.
  \item
    Since all $x_i$'s in Definition~\ref{def:composition} are promised 
    to have the same length, we can consider a padded version $\tilde 
    L$ of the language~$L$ in order to filter the input instances of 
    length~$n$ of the original~$L$ into a polynomial number of 
    equivalence classes.
    Each input length of $\tilde L$ in some interval $[p_1(n),p_2(n)]$ 
    corresponds to one equivalence class of length-$n$ instances 
    of~$L$.
    So long as $\tilde L$ remains $\NP$-hard or $\coNP$-hard, it is 
    sufficient to consider a composition from $\tilde L$ into $\Pi$.
    \textcite[Definition~4]{crosscomposition} formalize this approach.
  \item
    The composition algorithm can also use randomness, as long as the 
    overall probability gap of the concatenation of composition and 
    kernelization is not negligible.
  \item
    In the case that $L$ is $\NP$-hard, \textcite{FortnowSanthanam} and 
    \textcite{DellVanMelkebeek} prove that the composition algorithm can also be 
    a $\coNP$-algorithm or even a $\coNP$ oracle communication game in order to 
    get the collapse.
    Interestingly, this does not seem to follow from Drucker's proof nor from 
    the proof presented here, and it seems to require the full completeness 
    condition for the OR-composition.
    \textcite{kratsch2012co,pointlinecover} exploit these variants of the 
    composition framework to prove kernel lower bounds.
\end{enumerate}

\section{Preliminaries}
\def\RemainingSet{R}%

For any set $\RemainingSet\subseteq \zo^*$ and any $\ell\in\N$, we 
write $\RemainingSet_\ell \doteq \RemainingSet \cap \zo^\ell$ for the 
set of all length-$\ell$ strings inside of $\RemainingSet$.
For any $t\in\N$, we write $[t] \doteq \set{1,\dots,t}$.
For a set $V$, we write $\binom{V}{\leq t}$ for the set of all subsets 
$x\subseteq V$ that have size at most $t$.
We will work over a finite alphabet, usually $\Sigma=\zo$.
For a vector $a\in\Sigma^t$, a number $j\in[t]$, and a value 
$y\in\Sigma$, we write $a\vert_{j\leftarrow y}$ for the string that 
coincides with~$a$ except in position~$j$, where it has value~$y$.
For background in complexity theory, we defer to the book by 
\textcite{AB-book}.
We assume some familiarity with the complexity classes $\NP$ and 
$\coNP$ as well as their non-uniform versions $\NPpoly$ and 
$\coNPpoly$.

\subsection{Distributions and Randomized Mappings}

A \emph{distribution} on a finite ground set $\Omega$ is a function 
$\distD\colon \Omega\to[0,1]$ with
${\sum_{\omega\in\Omega}\distD(\omega)=1}$.
The \emph{support} of $\distD$ is the set 
$\supp\distD=\setc*{\omega\in\Omega}{\distD(\omega)>0}$.
The \emph{uniform distribution} $\Un_{\Omega}$ on $\Omega$ is the 
distribution with $\Un_{\Omega}(\omega)=\frac{1}{|\Omega|}$ for all 
$\omega\in\Omega$.
We often view distributions as \emph{random variables}, that is, we 
may write $f(\distD)$ to denote the distribution $\distD'$ that first 
produces a sample $\omega\sim\distD$ and then outputs $f(\omega)$, where 
$f\colon\Omega\to\Omega'$.
We use any of the following notations:
\begin{align*}
  \distD'(\omega')
  =\Pr ( f(\distD)=\omega' )
  =\Pr_{\omega\sim\distD} ( f(\omega)=\omega' )
  =\sum_{\omega\in\Omega} \distD(\omega) \cdot 
  \Pr(f(\omega)=\omega')\,.
\end{align*}
The last term $\Pr(f(\omega)=\omega')$ in this equation is either $0$ 
or $1$ if $f$ is a deterministic function, but we will also allow $f$ 
to be a \emph{randomized mapping}, that is, $f$ has access to some 
``internal'' randomness.
This is modeled as a function $f\colon\Omega\times\zo^r\to\Omega'$ for some 
$r\in\N$, and we write $f(\distD)$ as a short-hand for $f(\distD,\Un_{\zo^r})$.
That is, the internal randomness consists of a sequence of independent 
and fair coin flips.

\subsection{Statistical Distance}
The \emph{statistical distance} $d(X,Y)$ between two distributions $X$ 
and $Y$ on $\Omega$ is defined as
\begin{align}\label{def:statdist}
  d(X,Y)
  &=
  \max_{T\subseteq\Omega}
  \big|\Pr(X \in T) - \Pr(Y \in T)\big|\,.
\end{align}
The statistical distance between $X$ and $Y$ is a number in $[0,1]$, 
with $d(X,Y)=0$ if and only if $X=Y$ and $d(X,Y)=1$ if and only if the 
support of $X$ is disjoint from the support of $Y$.
It is an exercise to show the standard equivalence between the 
statistical distance and the \mbox{$1$-norm}:
\begin{align*}
  d(X,Y)
  &=
  \frac{1}{2}\cdot
  \big\|X - Y \big\|_1
  = \frac{1}{2}\sum_{\omega\in\Omega} \big|\Pr(X=\omega) - 
  \Pr(Y=\omega)\big|\,.
\end{align*}

\subsection{The Statistical Distance Problem}
For $\Un=\Un_{\zo^n}$ and $0\leq \delta < \Delta \leq 1$, let $\SD$ be 
the following promise problem:
\begin{description}
  \item[\normalfont\itshape yes-instances:]
    Pairs of circuits $C,C' \colon \zo^n \to \zo^*$ so that $d(C(\Un),C'(\Un)) 
    \geq \Delta$.
  \item[\normalfont\itshape no-instances:]
    Pairs of circuits $C,C' \colon \zo^n \to \zo^*$ so that $d(C(\Un),C'(\Un)) 
    \leq \delta$.
\end{description}

The statistical distance problem is not known to be polynomial-time 
computable, and in fact it is not believed to be.
On the other hand, the problem is also not believed to be \NP-hard 
because the problem is computationally easy in the following sense.
\begin{theorem}[\textcite{XiaoThesis} + \textcite{Adleman}]%
  \label{theorem:SD easy}%
  \mbox{}\\
  If $\delta<\Delta$ are constants, we have
  $\SD \in \Big(\NPpoly \;\cap\; \coNPpoly\Big)\,.$\\
  Moreover, the same holds when $\delta=\delta(n)$ and 
  $\Delta=\Delta(n)$ are functions of the input length that satisfy 
  $\Delta-\delta\geq\frac{1}{\poly(n)}$.
\end{theorem}
This is the only fact about the SD-problem that we will use in this 
paper.

Slightly stronger versions of this theorem are known:
For example, \textcite[p.~144ff]{XiaoThesis} proves that $\SD\in 
\AM\cap\coAM$ holds.
In fact, Theorem~\ref{theorem:SD easy} is established by combining his 
theorem with the standard fact that $\AM \subseteq \NPpoly$, i.e., 
that Arthur--Merlin games can be derandomized with polynomial advice 
\parencite{Adleman}.
Moreover, when we have the stronger guarantee that $\Delta^2>\delta$ 
holds, then $\SD$ can be solved using \emph{statistical zero-knowledge 
  proof systems} \parencite{SV03,GV11}.
Finally, if $\Delta=1$, the problem can be solved with \emph{perfect} 
zero-knowledge proof systems \parencite[Proposition~5.7]{SV03}.
Using these stronger results whenever possible gives slightly stronger 
complexity collapses in the main theorem.

\section{Ruling out OR-compressions}
\label{sec:proof}

In this section we prove Theorem~\ref{thm main simple}: Any language 
$L$ that has a relaxed OR-compression is in $\coNPpoly \cap \NPpoly$.
We rephrase the theorem in a form that reveals the precise inequality 
between the error probabilities and the compression ratio needed to 
get the complexity consequence.
\begin{theorem}[$\epsilon t$-compressive version of Drucker's theorem]
  \label{thm main}
  Let $L,L'\subseteq\zo^*$ be languages and $e_s,e_c\in[0,1]$ be some 
  constants denoting the error probabilities.
  Let $t=t(n)>0$ be a polynomial and $\epsilon>0$.
  Let
  \begin{align}\label{eq:A type}
    A \colon \binom{\zo^n}{\leq t}\to\zo^{\epsilon t}
  \end{align}
  be a randomized
  $\Ppoly$-algorithm such that,
  for all $x \in \binom{\zo^n}{\leq t}$,
  \begin{itemize}
    \item if $|x\cap L|=0$, then $A(x)\in \overline{L'}$ holds with 
      probability $\geq 1-e_s$, and
    \item if $|x\cap L|=1$, then $A(x)\in L'$ holds with probability 
      $\geq 1-e_c$.
  \end{itemize}

  If $e_s+e_c < 1 - \sqrt{(2\ln 2) \epsilon}$,
  then $L\in\NPpoly \cap \coNPpoly$.
\end{theorem}
This is Theorem~7.1 in \textcite{Druckerfull}.
However, there are two noteworthy differences:
\begin{enumerate}
\item
  Drucker obtains complexity consequences even when
  ${e_s+e_c}< 1 - \sqrt{(\ln 2/2) \epsilon}$ holds, which makes his 
  theorem more general.
  The difference stems from the fact that we optimized the proof in 
  this section for simplicity and not for the optimality of the bound.
  He also obtains complexity consequences under the (incomparable) 
  bound
  ${e_s+e_c}< 2^{-\epsilon-3}$.
  Using the slightly more complicated setup of~\S\ref{sec:proof2}, we 
  would be able to achieve both of these bounds.
\item
  To get a meaningful result for OR-compression of \NP-complete 
  problems, we need the complexity consequence
  $L\in\NPpoly \cap \coNPpoly$ rather than just $L\in\NPpoly$.
  To get the stronger consequence, Drucker relies on the fact that the 
  statistical distance problem $\SD$ has statistical zero knowledge 
  proofs.
  This is only known to be true when $\Delta^2>\delta$ holds, which translates 
  to the more restrictive assumption
  $(e_s+e_c)^2 < 1-\sqrt{(\ln 2/2)\epsilon}$ in his theorem.
  We instead use Theorem~\ref{theorem:SD easy}, which does not go 
  through statistical zero knowledge and proves more directly that $\SD$ is in 
  $\NPpoly\cap\coNPpoly$ whenever $\Delta>\delta$ holds.
  Doing so in Drucker's paper immediately improves all of his 
  $L\in\NPpoly$ consequences to $L\in\NPpoly\cap\coNPpoly$.
\end{enumerate}
To obtain Theorem~\ref{thm main simple}, the basic version of Drucker's theorem, 
as a corollary of Theorem~\ref{thm main}, none of these differences matter.
This is because we could choose $\epsilon>0$ to be sufficiently smaller in the 
proof of Theorem~\ref{thm main simple}, which we provide now before we turn to 
the proof of Theorem~\ref{thm main}.
\begin{proof}[of Theorem~\ref{thm main simple}]
  Let $A$ be the algorithm assumed in Theorem~\ref{thm main simple}, 
  and let $C\geq 2$ be large enough so that the output size of $A$ is 
  bounded by $t^{1-1/C}\cdot C \cdot n^C$.
  We transform~$A$ into an algorithm as required for Theorem~\ref{thm 
  main}.
  Let $\epsilon>0$ be a small enough constant so that $e_s+e_c < 1-
  \sqrt{(2\ln 2) \epsilon}$.
  Moreover, let $t(n)$ be a large enough polynomial so that 
  $(t(n))^{1-1/C}\cdot C\cdot n^C < \epsilon t(n)$ holds.
  Then we restrict $A$ to a family of functions
  $A_n\colon\binom{\zo^n}{\leq t(n)}\to\zo^{<\epsilon t(n)}$.
  Now a minor observation is needed to get an algorithm of the form~\eqref{eq:A 
    type}:
  The set $\zo^{<\epsilon t}$ can be efficiently encoded in 
  $\zo^{\epsilon t}$ (which changes the output language from $L'$ to 
  some $L''$).
  Thus we constructed a family~$A_n$ as required by Theorem~\ref{thm 
    main}, which proves the claim.
\end{proof}

\subsection{ORs are sensitive to Yes-instances}

The semantic property of relaxed OR-compressions is that they are 
``$L$-sensitive'': They show a dramatically different behavior for 
all-no input sets vs.\ input sets that contain a single yes-instance 
of~$L$.
The following simple fact is the only place in the overall proof where 
we use the soundness and completeness properties of~$A$.
\begin{lemma}\label{lemma:yes}
  For all distributions $X$ on $\binom{\overline{L}}{<t}$ and all 
  $v\in L$, we have
  \begin{align}\label{eq:yes}
    d\Big( A(X) \;,\; A(X \cup \{v\} ) \Big) \geq \Delta \doteq 
    1-(e_s+e_c)\,.
  \end{align}
\end{lemma}
\begin{proof}
  The probability that $A(X)$ outputs an element of $L'$ is at most 
  $e_s$, and similarly, the probability that $A(X\cup\{v\})$ outputs 
  an element of $L'$ is at least $1-e_c$.
  By \eqref{def:statdist} with $T=L'$, the statistical distance 
  between the two distributions is at least $\Delta$.
\end{proof}
Despite the fact that relaxed OR-compressions are sensitive to the 
presence or absence of a yes-instance, we argue next that their 
behavior \emph{within} the set of no-instances is actually quite 
predictable.

\subsection{The average noise sensitivity of compressive maps is 
  small}

Relaxed OR-compressions are in particular compressive maps.
The following lemma says that the average noise sensitivity of any 
compressive map is low.
Here, ``average noise sensitivity'' refers to the difference in the 
behavior of a function when the input is subject to random noise; in 
our case, we change the input in a single random location and notice 
that the behavior of a compressive map does not change much.
\begin{lemma}\label{lemma:pinsker}
  Let $t\in\N$, let $X$ be the uniform distribution on $\zo^t$, and 
  let $\epsilon>0$.
  Then, for all randomized mappings $f\colon\zo^t\to\zo^{\epsilon t}$, we have
  \begin{align}\label{eq:noise sensitivity}
    \Exp_{j\sim\Un_{[t]}}
    \quad
    d
    \Big(
      f
      \big(
      X
      \vert_{j\leftarrow 0}
      \big)
      \;,\;
      f
      \big(
      X
      \vert_{j\leftarrow 1}
      \big)
    \Big)
    \quad
    &\leq
    \delta \doteq
    \sqrt{2 \ln 2 \cdot \epsilon}
    \,.
  \end{align}
\end{lemma}
We defer the purely information-theoretic and mechanical proof of this 
lemma to \S\ref{sec:info1}.
In the special case where $f\colon\zo^t\to\zo$ is a Boolean function, the 
left-hand side of~\eqref{eq:noise sensitivity} coincides with the usual 
definition of the average noise sensitivity.

We translate Lemma~\ref{lemma:pinsker} to our relaxed 
OR-compression~$A$ as follows.
\begin{lemma}\label{lemma:selector exists}
  Let $A\colon\binom{\zo^n}{\leq t}\to\zo^{\epsilon t}$.
  For all $e\in\binom{\zo^n}{t}$, there exists $v\in e$ so that
  \begin{align}\label{eq:selector}
    d
    \Big(
      A
      \big(
      \Un_{2^e}
      \setminus\{v\}
      \big)
      \;,\;
      A
      \big(
      \Un_{2^e}
      \cup\{v\}
      \big)
    \Big)
    &\leq
    \delta
    \,.
  \end{align}
\end{lemma}
Here $\Un_{2^e}$ samples a subset of~$e$ uniformly at random.
Note that we replaced the expectation over $j$ from \eqref{eq:noise 
  sensitivity} with the mere existence of an element~$v$ in 
\eqref{eq:selector} since this is all we need; the stronger property 
also holds.
\begin{proof}
  To prove the claim, let $v_1,\dots,v_t$ be the elements of $e$ in 
  lexicographic order.
  For $b\in\zo^t$, let $g(b)\subseteq e$ be such that $v_i\in g$ holds 
  if and only if $b_i=1$.
  We define the randomized mapping $f\colon\zo^t\to\zo^{\epsilon t}$ as follows:
  \begin{align*}
    f(b_1,\dots,b_t) \doteq A\Big(g(b)\Big)
    \,.
  \end{align*}
  Then $f(X|_{j\leftarrow 0})=A(\Un_{2^e}\setminus\{v_j\})$
  and $f(X|_{j\leftarrow 1})=A(\Un_{2^e}\cup\{v_j\})$.
  The claim follows from Lemma~\ref{lemma:pinsker} with $v\doteq v_j$ 
  for some $j$ that minimizes the statistical distance in 
  \eqref{eq:noise sensitivity}.
\end{proof}
This lemma suggest the following tournament idea.
We let $V=\overline{L}_n$ be the set of no-instances, and we let them 
compete in matches consisting of $t$ players each.
That is, a match corresponds to a hyperedge $e\in\binom{V}{t}$ of 
size~$t$ and every such hyperedge is present, so we are looking at a 
complete $t$-uniform hypergraph.
We say that a player $v\in e$ is ``selected'' in the hyperedge $e$ if 
the behavior of $A$ on $\Un_{2^e}\setminus\{v\}$ is not very different 
from the behavior of $A$ on $\Un_{2^e}\cup\{v\}$, that is, if 
\eqref{eq:selector} holds.
The point of this construction is that $v$ being selected proves that 
$v$ must be a no-instance because \eqref{eq:yes} does not hold.
We obtain a ``selector'' function $S\colon\binom{V}{t}\to V$ that, given $e$, 
selects an element $v=S(e)\in e$.
We call~$S$ a \emph{hypergraph tournament} on~$V$.

\subsection{Hypergraph tournaments have small dominating sets}

Tournaments are complete directed graphs, and it is well-known that 
they have dominating sets of logarithmic size.
A straightforward generalization applies to hypergraph tournaments 
$S\colon\binom{V}{t}\to V$.
We say that a set $g\in\binom{V}{t-1}$ \emph{dominates} a vertex $v$ 
if $v\in g$ or $S(g\cup\{v\})=v$ holds.
A set $\domset\subseteq \binom{V}{t-1}$ is a \emph{dominating set} of 
$S$ if all vertices $v\in V$ are dominated by at least one element in 
$\domset$.
\begin{lemma}\label{lem:domset}
  Let $V$ be a finite set, and let $S\colon\binom{V}{t}\to V$ be a hypergraph 
  tournament.

  Then $S$ has a dominating set $\domset\subseteq\binom{V}{t-1}$ of 
  size at most $t \log|V|$.
\end{lemma}
\begin{proof}
  We construct the set $\domset$ inductively.
  Initially, it has $k=0$ elements.
  After the $k$-th step of the construction, we will preserve the 
  invariant that $\domset$ is of size exactly~$k$ and that 
  $|\RemainingSet| \leq (1-1/t)^k \cdot |V|$ holds, where
  $R$ is the set of vertices that are not yet dominated, that is,
  \begin{align*}
    \RemainingSet &=
    \setc[\Big]{
      v \in V
    }{
      \text{%
        $v\not\in g$
        and
        $S(g\cup\{v\})\neq v$
        holds for all
        $g\in \domset$%
      }
    }\,.
  \end{align*}
  If $0<|R|<t$, we can add an arbitrary edge $g^*\in\binom{V}{t-1}$ 
  with $R\subseteq g^*$ to $\domset$ to finish the construction.
  Otherwise, the following averaging argument, shows that there is an 
  element $g^*\in\binom{\RemainingSet}{t-1}$ that dominates at least a 
  $1/t$-fraction of elements $v\in\RemainingSet$:
  \begin{align*}
    \frac{1}{t}
    =
    \Exp_{e\in\binom{\RemainingSet}{t}}
    \Pr_{v\in e} \Big(S(e) = v\Big)
    &=
    \Exp_{g\in\binom{\RemainingSet}{t-1}}
    \Pr_{v\in \RemainingSet-g} \Big(S(g\cup\{v\}) = v\Big)
    \,.
  \end{align*}
  Thus, the number of elements of $\RemainingSet$ left undominated by 
  $g^*$ is at most $(1-1/t) \cdot |\RemainingSet|$, so the inductive 
  invariant holds.
  Since $(1-1/t)^{k} \cdot |V| \leq \exp({-k/t}) \cdot |V| < 1$ for
  $k=t\log |V|$, we have $\RemainingSet=\emptyset$ after $k\leq t\log 
  |V|$ steps of the construction, and in particular, $\domset$ has at 
  most $t \log|V|$ elements.
\end{proof}

\subsection{Proof of the main theorem: Reduction to statistical 
  distance}
\begin{proof}[of Theorem~\ref{thm main}]
  We describe a deterministic $\Ppoly$ reduction from $L$ to the 
  statistical distance problem $\SD$ with $\Delta=1-(e_s+e_c)$ and 
  $\delta=\sqrt{( 2\ln 2 ) \epsilon}$.
  The reduction outputs the conjunction of polynomially many instances 
  of $\SD$.
  Since $\SD$ is contained in the intersection of $\NPpoly$ and 
  $\coNPpoly$ by Theorem~\ref{theorem:SD easy},
  and since this intersection is closed under taking polynomial 
  conjunctions, we obtain $L\in\NPpoly\cap\coNPpoly$.
  Thus it remains to find such a reduction.
  To simplify the discussion, we describe the reduction in terms of an 
  algorithm that solves~$L$ and uses $\SD$ as an oracle.
  However, the algorithm only makes non-adaptive queries at the end of 
  the computation and accepts if and only if all oracle queries 
  accept; this corresponds to a reduction that maps an instance of $L$ 
  to a conjunction of instances of $\SD$ as required.

  To construct the advice at input length~$n$, we use 
  Lemma~\ref{lemma:selector exists} with $t=t(n)$ to obtain a 
  hypergraph tournament~$S$ on~$V=\overline{L}_n$, which in turn gives 
  rise to a small dominating set $\domset\subseteq \binom{V}{t-1}$ by 
  Lemma~\ref{lem:domset}.
  We remark the triviality that if $|V|\leq t=\poly(n)$, then we can 
  use~$V$, the set of all no-instances of~$L$ at this input length, as 
  the advice.
  Otherwise, we define the hypergraph tournament~$S$ for all $e\in 
  \binom{V}{t}$ as follows:
  \begin{align*}
    S(e)
    &\doteq
    \min
    \setc*{
      v\in e
    }{
      d\paren[\Big]{
        A(\Un_{2^e} \setminus \{v\})
        \;,\;
        A(\Un_{2^e} \cup\{v\})
      }
      \leq \delta
    }\,.
  \end{align*}
  By Lemma~\ref{lemma:selector exists}, the set over which the minimum 
  is taken is non-empty, and thus~$S$ is well-defined.
  Furthermore, the hypergraph tournament has a dominating 
  set~$\domset$ of size at most $t n$ by Lemma~\ref{lem:domset}.
  As advice for input length~$n$, we choose this set~$\domset$.
  Now we have $v\in\overline{L}$ if and only if $v$ is dominated by 
  $\domset$.
  The idea of the reduction is to efficiently check the latter 
  property.

  The algorithm works as follows:
  Let $v\in\zo^n$ be an instance of~$L$ given as input.
  If $v\in g$ holds for some $g\in\domset$, the algorithm rejects~$v$ 
  and halts.
  Otherwise, it queries the SD-oracle on the instance $(A(\Un_{2^g}), 
  A(\Un_{2^g} \cup \{v\}))$ for each $g\in\domset$.
  If the oracle claims that all queries are yes-instances, our 
  algorithm accepts, and otherwise, it rejects.

  First note that distributions of the form
  $
    A
    \big(
    \Un_{2^g}
    \big)
  $
  and
  $
    A
    \big(
    \Un_{2^g}
    \cup\{v\}
    \big)
  $
  can be be sampled by using polynomial-size circuits, and so they 
  form syntactically correct instances of the SD-problem:
  The information about $A$, $g$, and $v$ is hard-wired into these 
  circuits, the input bits of the circuits are used to produce a 
  sample from $\Un_{2^g}$, and they serve as internal randomness 
  of~$A$ in case~$A$ is a randomized algorithm.

  It remains to prove the correctness of the reduction.
  If $v\in L$, we have for all $g\in\domset\subseteq {\overline L}$ 
  that $v\not\in g$ and that the statistical distance of the query corresponding 
  to $g$ is at least $\Delta=1-(e_s+e_c)$ by Lemma~\ref{lemma:yes}.
  Thus all queries that the reduction makes satisfy the promise of the 
  SD-problem and the oracle answers the queries correctly, leading our 
  reduction to accept.
  On the other hand, if $v\not\in L$, then, since $\domset$ is a 
  dominating set of~$\overline{L}$ with respect to the hypergraph 
  tournament~$S$, there is at least one $g\in \domset$ so that $v\in 
  g$ or $S(g\cup\{v\})=v$ holds.
  If $v\in g$, the reduction rejects.
  The other case implies that the statistical distance between 
  $A(\Un_{2^g})$ and $A(\Un_{2^g}\cup\{v\})$ is at most $\delta$.
  The query corresponding to this particular~$g$ therefore satisfies 
  the promise of the SD-problem, which means that the oracle answers 
  correctly on this query and our reduction rejects.
\end{proof}

\subsection{Information-theoretic arguments}
\label{sec:info1}

We now prove Lemma~\ref{lemma:pinsker}.
The proof uses the \emph{Kullback--Leibler divergence} as an 
intermediate step.
Just like the statistical distance, this notion measures how similar 
two distributions are, but it does so in an information-theoretic way 
rather than in a purely statistical way.
In fact, it is well-known in the area that the Kullback--Leibler 
divergence and the mutual information are almost interchangeable in a 
certain sense.
We prove a version of this paradigm formally in Lemma~\ref{lemma:KL 
bound} below; then we prove Lemma~\ref{lemma:pinsker} by bounding the 
statistical distance in terms of the Kullback--Leibler divergence 
using standard inequalities.

We introduce some basic information-theoretic notions.
The \emph{Shannon entropy} $H(X)$ of a random variable~$X$ is
\begin{align*}
  H(X) = \Exp_{x\sim X} \log\left(\frac{1}{\Pr(X=x)}\right)\,.
\end{align*}
The conditional Shannon entropy $H(X|Y)$ is
\begin{align*}
  H(X|Y) &= \Exp_{y\sim Y} H(X|Y=y)\\
  &= \Exp_{y\sim Y} \sum_x \Prc{X=x}{Y=y} \cdot
  \log\left(\frac{1}{\Prc{X=x}{Y=y}}\right)\,.
\end{align*}
The \emph{mutual information} between~$X$ and~$Y$ is 
$I(X:Y)=H(X)-H(X|Y)=H(Y)-H(Y|X)$.
Note that $I(X:Y) \leq\log|\supp X|$, where $|\supp X|$ is the size of 
the support of $X$.
The \emph{conditional mutual information} can be defined by the 
\emph{chain rule of mutual information}
$I(X:Y\;|\;Z)=I(X:YZ) - I(X:Z)$.
If $Y$ and $Z$ are independent, then a simple calculation reveals that 
$I(X:Y)\leq I(X:Y\;|\;Z)$ holds.

We now establish a bound on the Kullback--Leibler divergence.
The application of Lemma~\ref{lemma:pinsker} only uses $\Sigma=\zo$.
The proof does not become more complicated for general~$\Sigma$, and 
we will need the more general version later in this paper.
\begin{lemma}\label{lemma:KL bound}
  Let $t\in\N$ and let $X_1,\dots,X_t$ be independent distributions on some 
  finite set $\Sigma$, and let $X=X_1,\dots,X_t$.
  Then, for all randomized mappings $f\colon\Sigma^t\to\zo^*$,
  we have the following upper bound on the expected value of the 
  Kullback--Leibler divergence:  \begin{align*}
    \Exp_{j\sim\Un_{[t]}}
    \Exp_{x\sim X_j}
    \KLdiv
    \Big(
    f
    \big(
    X
    \big)
    \;||\;
    f
    \big(
    X
    \vert_{j\leftarrow x}
    \big)
    \Big)
    \leq
    \frac{1}{t}\cdot I\big(f(X) : X\big)\,.
  \end{align*}
\end{lemma}
\begin{proof}
  The result follows by a basic calculation with entropy notions.
  The first equality is the definition of the Kullback--Leibler 
  divergence, which we rewrite using the logarithm rule
  $\log(a/b)=\log(1/b)-\log(1/a)$ and the linearity of expectation:
  \begin{align*}
    &\Exp_{j}
    \Exp_{x}
    \KLdiv
    \Big(
      f
      \big(
      X
      \big)
      \;||\;
      f
      \big(
      X
      \vert_{j\leftarrow x}
        \big)
    \Big)\\
    &=
    \Exp_{j}
    \Exp_{x}
    \sum_z \log \left(\frac{
      \Pr(f(
      X
      \vert_{j\leftarrow x}
        )=z)
      }{
      \Pr(f( X )=z)
      }\right)
    \cdot
      \Pr\Big(f(
      X
      \vert_{j\leftarrow x}
      )=z\Big)\\
    &=
    \Exp_{j}
    \sum_z \log \left(\frac{
        1
      }{
      \Pr(f( X )=z)
      }\right)
    \cdot
    \Exp_{x}
      \Pr\Big(f(
      X
      \vert_{j\leftarrow x}
      )=z\Big)\\
    &\quad -
    \Exp_{j}
    \Exp_{x}
    \sum_z \log \left(\frac{
        1}{
      \Pr(f(
      X
      \vert_{j\leftarrow x}
        )=z)
      }
      \right)
    \cdot
      \Pr\Big(f(
      X
      \vert_{j\leftarrow x}
      )=z\Big)
    \,.
    \end{align*}
    As $
    \Exp_{x}
    \Pr\big(f(
    X
    \vert_{j\leftarrow x}
    )=z\big)
    = \Pr(f(X)=z)$, both terms of the sum above are entropies, and we 
    can continue the calculation as follows:
    \begin{align*}
    \dots&=
    H(f(X))
    -
    \Exp_{j}
    \Exp_{x}
    H(f(X) | X_j=x)
    &\text{(definition of entropy)}\\
    &=
    H(f(X))
    -
    \Exp_{j}
    H(f(X) | X_j)
    &\text{(definition of conditional entropy)}\\
    &=
    \Exp_{j}
    I(f(X):X_j)
    &\text{(definition of mutual information)}\\
    &\leq\frac{1}{t} \cdot \sum_{j\in[t]}
    I\Big(f(X):X_j \;\big\vert\; X_1 \dots X_{j-1}\Big)
    &\text{(by independence of $X_j$'s)}\\
    &= \frac{1}{t} \cdot I(f(X):X)\,.
    &\text{(chain rule of mutual information)}
  \end{align*}
\end{proof}
We now turn to the proof of Lemma~\ref{lemma:pinsker}, where 
we bound the statistical distance in terms of the Kullback--Leibler 
divergence.
\begin{proof}[of Lemma~\ref{lemma:pinsker}]
  We observe that $I(f(X):X) \leq \log |\supp f(X)| \leq \epsilon t$, 
  and so we are in the situation of Lemma~\ref{lemma:KL bound} with 
  $\Sigma=\zo$.
  We first apply the triangle inequality to the left-hand side 
  of~\eqref{eq:noise sensitivity}.
  Then we use Pinsker's inequality \parencite[Lemma~11.6.1]{cover2012elements} 
  to bound the statistical distance in terms of the Kullback--Leibler 
  divergence, which we can in turn bound by~$\epsilon$ using Lemma~\ref{lemma:KL 
    bound}.
  \begin{align*}
    &\Exp_{j\sim\Un_{[t]}}
    d
    \Big(
      f
      \big(
      X
      \vert_{j\leftarrow 0}
      \big)
      \;,\;
      f
      \big(
      X
      \vert_{j\leftarrow 1}
      \big)
    \Big)\\
    &\leq
    \Exp_{j\sim\Un_{[t]}}
    d
    \Big(
      f
      \big(
      X
      \big)
      \;,\;
      f
      \big(
      X
      \vert_{j\leftarrow 0}
      \big)
    \Big)
    & \text{(triangle inequality)}\\
    &+
    \Exp_{j\sim\Un_{[t]}}
    d
    \Big(
      f
      \big(
      X
      \big)
      \;,\;
      f
      \big(
      X
      \vert_{j\leftarrow 1}
      \big)
    \Big)\\
    &=
    2 \cdot
    \Exp_{j\sim\Un_{[t]}}
    \Exp_{x\sim X_j}
    d
    \Big(
      f
      \big(
      X
      \big)
      \;,\;
      f
      \big(
      X
      \vert_{j\leftarrow x}
      \big)
    \Big)\\
    &\leq
    2\cdot
    \Exp_{j}
    \Exp_{x}
    \sqrt{
      \frac{\ln 2}{2}\cdot
    \KLdiv
    \Big(
      f
      \big(
      X
      \big)
      \;||\;
      f
      \big(
      X
      \vert_{j\leftarrow x}
      \big)
    \Big)}
    & \text{(Pinsker's inequality)}\\
    &\leq
    2 \cdot
    \sqrt{
      \frac{\ln 2}{2}\cdot
    \Exp_{j}
    \Exp_{x}
    \KLdiv
    \Big(
      f
      \big(
      X
      \big)
      \;||\;
      f
      \big(
      X
      \vert_{j\leftarrow x}
      \big)
    \Big)}
    & \text{(Jensen's inequality)}\\
    &\leq 2\cdot \sqrt{\frac{\ln 2}{2} \cdot \epsilon}=\delta\,. & 
    \text{(Lemma~\ref{lemma:KL bound})}
  \end{align*}
  The equality above uses the fact that $X_j$ is the uniform 
  distribution on~$\zo$.
\end{proof}

\section{Extension: Ruling out OR-compressions of size 
\texorpdfstring{$\bm{O(t \log t)}$}{O(t log t)}}
\label{sec:proof2}

In this section we tweak the proof of Theorem~\ref{thm main} so 
that it works even when the~$t$ instances of $L$ are mapped to an 
instance of $L'$ of size at most $O(t\log t)$.
The drawback is that we cannot handle positive constant error 
probabilities for randomized relaxed OR-compression anymore.
For simplicity, we restrict ourselves to \emph{deterministic} relaxed 
OR-compressions of size $O(t\log t)$ throughout this section.

\begin{theorem}[$O(t \log t)$-compressive version of Drucker's 
  theorem]\label{thm main tlogt}
  \mbox{}\\
  Let $L,L'\subseteq\zo^*$ be languages.
  Let $t=t(n)>0$ be a polynomial.
  Assume there exists a $\Ppoly$-algorithm
  \begin{align*}
    A \colon \binom{\zo^n}{\leq t}\to\zo^{O(t \log t)}
  \end{align*}
  such that, for all $x \in \binom{\zo^n}{\leq t}$,
  \begin{itemize}
    \item if $|x\cap L|=0$, then $A(x)\in \overline{L'}$, and
    \item if $|x\cap L|=1$, then $A(x)\in L'$.
  \end{itemize}

  Then $L\in\NPpoly \cap \coNPpoly$.
\end{theorem}
This is Theorem~7.1 in \textcite{Druckerfull}.
The main reason why the proof in \S\ref{sec:proof} breaks down for 
compressions to size $\epsilon t$ with $\epsilon=O(\log t)$ is that 
the bound on the statistical distance in Lemma~\ref{lemma:pinsker} 
becomes trivial.
This happens already when $\epsilon \geq \frac{1}{2 \ln 2}\approx 
0.72$.
On the other hand, the bound that Lemma~\ref{lemma:KL bound} gives for 
the Kullback--Leibler divergence remains non-trivial even for 
$\epsilon=O(\log t)$.
To see this, note that the largest possible divergence
between $f(X)$ and $f(X\vert_{j\leftarrow x})$, that is, the 
  divergence without the condition on the mutual information between 
  $f(X)$ and $X$,
is $t \cdot \log |\Sigma|$, and the bound that Lemma~\ref{lemma:KL 
  bound} yields for $\epsilon=O(\log t)$ is logarithmic in that.

Inspecting the proof of Lemma~\ref{lemma:pinsker}, we realize that the 
loss in meaningfulness stems from Pinsker's inequality, which becomes 
trivial in the parameter range under consideration.
Luckily, there is a different inequality between the statistical 
distance and the Kullback--Leibler divergence, Vajda's inequality, 
that still gives a non-trivial bound on the statistical distance when 
the divergence is $\geq\frac{1}{2\ln 2}$.
The inequality works out such that if the divergence is logarithmic, 
then the statistical distance is an inverse polynomial away from~$1$.
We obtain the following analogue to Lemma~\ref{lemma:pinsker}.
\begin{lemma}\label{lemma:vajda}
  Let $t\in\N$ let $X_1,\dots,X_t$ be independent uniform 
  distributions on some finite set $\Sigma$, and write 
  $X=X_1,\dots,X_t$.
  Then, for all randomized mappings $f\colon\Sigma^t\to\zo^*$ with $I\big(f(X) : 
  X\big)\leq O(t \cdot \log t)$, we have
  \begin{align}\label{eq:vajda}%
    \Exp_{j\sim\Un_{[t]}}
    \Exp_{x\sim X_j}
    d
    \Big(
    f
    \big(
    X
    \vert_{X_j\neq x}
    \big)
    \;,\;
    f
    \big(
    X
    \vert_{X_j=x}
    \big)
    \Big)
    \leq
    1-\frac{1}{\poly(t)}+\frac{1}{|\Sigma|}
    \,.
  \end{align}
\end{lemma}
The notation $X\vert_{X_j\neq x}$ refers to the random variable that 
samples $x_i\sim X_i=\Un_{\Sigma}$ independently for each $i\neq j$ as 
usual, and that samples $x_j$ from the distribution $X_j$ conditioned 
on the event that $X_j\neq x$, that is, the distribution 
$\Un_{\Sigma\setminus\{a\}}$.
The notation $X\vert_{X_j=x}=X\vert_{j\leftarrow x}$ is as before, 
that is, $x_j=x$ is fixed.

We defer the proof of the lemma to the end of this section and discuss 
now how to use it to obtain the stronger result for $O(t \log t)$ 
compressions.
First note that we could not have directly used 
Lemma~\ref{lemma:vajda} in place of Lemma~\ref{lemma:pinsker} in the 
proof of the main result, Theorem~\ref{thm main}.
This is because for $\Sigma=\zo$, the right-hand side 
of~\eqref{eq:vajda} becomes bigger than $1$ and thus trivial.
In fact, this is the reason why we formulated Lemma~\ref{lemma:KL 
bound} for general $\Sigma$.
We need to choose $\Sigma$ with $|\Sigma| = \poly(t)$ large enough to 
get anything meaningful out of~\eqref{eq:vajda}.

\subsection{A different hypergraph tournament}
To be able to work with larger $\Sigma$, we need to define the 
hypergraph tournament in a different way; not much is changing on a 
conceptual level, but the required notation becomes a bit less 
natural.
We do this as follows.
\begin{lemma}\label{lemma:tlogt selector}
  Let $A\colon\binom{\zo^n}{\leq t}\to\zo^{\epsilon t}$.
  There exists a large enough constant $C\in\N$ such that
  with $\Sigma=[t^C]$ we have:
  For all $e=e_1 \dotcup e_2 \dotcup \dots \dotcup e_{t} 
  \subseteq\zo^n$ with $|e_i|=|\Sigma|$,
  there exists an element $v\in e$ so that
  \begin{align}\label{eq:tlogt selector}
    \sdist
    \Big(
    A
    \big(
    X_e \vert_{v\not\in X_e}
    \big)
    \;,\;
    A
    \big(
    X_e \vert_{v\in X_e}
    \big)
    \Big)
    \leq
    1-\frac{1}{\poly(t)}
    \,,
  \end{align}
  where~$X_e$ is the distribution that samples the $t$-element set 
  $\set{\Un_{e_1},\dots,\Un_{e_t}}$, and $X_e\vert_{E}$ is the 
  distribution $X_e$ conditioned on the event~$E$.
\end{lemma}
For instance if $v\in e_1$, then
$X_e \vert_{v\not\in X_e}$ samples the $t$-element set $\set{
\Un_{e_1\setminus\{v\}},
\Un_{e_2},
\dots,
\Un_{e_t}
}$
and
$X_e \vert_{v\in X_e}$ samples the $t$-element set
$\set{
v,
\Un_{e_2},
\dots,
\Un_{e_t}
}$.
The proof of this lemma is analogous to the proof of 
Lemma~\ref{lemma:selector exists}.
\begin{proof}
  We choose $C$ as a constant that is large enough so that the 
  right-hand side of~\eqref{eq:vajda} becomes bounded 
  by~$1-1/\poly(t)$.
  Let $e_{ia}\in\zo^n$ for $i\in[t]$ and $a\in\Sigma$ be the 
  lexicographically $a$-th element of~$e_i$.
  We define the function $f\colon\Sigma^t \to \zo^{O(t\log t)}$ as follows:
  $f(a_1,\dots,a_t) \doteq A(e_{1a_1},\dots,e_{ta_t})$.
  Finally, we let the distributions $X_i$ be $X_i=\Un_{\Sigma}$ for 
  all $i\in[t]$.
  We apply Lemma~\ref{lemma:pinsker} to $f$ and obtain indices 
  $j\in[t]$ and $x\in \Sigma$ minimizing the statistical distance on 
  the left-hand side of~\eqref{eq:vajda}.
  Since $f(X_e\vert_{e_{jx}\not\in X_e})=A(X\vert_{X_j\neq x})$
  and $f(X_e\vert_{e_{jx}\in X_e})=A(X\vert_{X_j= x})$,
  we obtain the claim with $v\doteq e_{jx}$.
\end{proof}
\subsection{Proof of Theorem~\ref{thm main tlogt}}
\begin{proof}[of Theorem~\ref{thm main tlogt}]
  As in the proof of Theorem~\ref{thm main}, we construct a 
  deterministic $\Ppoly$ reduction from $L$ to a conjunction of 
  polynomially many instances of the statistical distance problem 
  $\SD$, but this time we let $D=1$ and $\delta=1-\frac{1}{\poly(t)}$ 
  be equal to the right-hand side of~\eqref{eq:tlogt selector}.
  Since there is a polynomial gap between $d$ and $D$, 
  Theorem\ref{theorem:SD easy} implies that $\SD$ is contained in the 
  intersection of $\NPpoly$ and $\coNPpoly$.
  Since the intersection is closed under polynomial disjunctions,
  we obtain $L\in\NPpoly\cap\coNPpoly$.
  Thus it remains to find such a reduction.

  To construct the advice at input length~$n$, we use 
  Lemma~\ref{lemma:tlogt selector} with $t=t(n)$, which guarantees that the 
  following hypergraph tournament $S\colon\binom{V}{|\Sigma| \cdot t}\to V$ with 
  $V=\overline{L}_n$ is well-defined:
  \begin{align*}
    S(e)
    &\doteq
    \min
    \setc*{
      v\in e
    }{
      d\paren[\Big]{
        A(X_e \vert_{v\not\in X_e})
        \;,\;
        A(X_e \vert_{v\in X_e})
      }
      \leq \delta
    }\,.
  \end{align*}
  We remark that if $|V|\leq |\Sigma|t=\poly(n)$, then we can use~$V$ 
  as the advice.
  Otherwise, the advice at input length~$n$ is the dominating set 
  $\domset\subseteq \binom{V}{\Sigma\cdot t-1}$ guaranteed by 
  Lemma~\ref{lem:domset}; in particular, its size is bounded by $t 
  \cdot|\Sigma|\cdot n=\poly(n)$.

  The algorithm for $L$ that uses $\SD$ as an oracle works as follows:  
  Let $v\in\zo^n$ be an instance of~$L$ given as input.
  If $v\in g$ holds for some $g\in\domset$, the reduction rejects~$v$ 
  and halts.
  Otherwise, for each $g\in\domset$, it queries the SD-oracle on the 
  instance
  $\big(
  A(X_e \vert_{v\not\in X_e}) \;,\; A(X_e \vert_{v\in X_e})
  \big)$
  with $e=g\cup\{v\}$.
  If the oracle claims that all queries are yes-instances, our 
  reduction accepts, and otherwise, it rejects.

  The correctness of this reduction is analogous to the proof 
  Theorem~\ref{thm main}:
  If $v\in L$, then Lemma~\ref{lemma:yes} guarantees that the 
  statistical distance of all queries is one, and so all queries will 
  detect this.
  If $v\in\overline{L}$, then since $\domset$ is a dominating set of 
  $S$, we have $v\in g$ or $S(g\cup\{v\})=v$ for some $g\in\domset$.
  The latter will be detected in the query corresponding to $g$ since 
  $\delta<D$.
  This completes the proof of the theorem.
\end{proof}

\subsection{Information-theoretic arguments}
\begin{proof}[of Lemma~\ref{lemma:vajda}]
  We use Vajda's inequality
  \parencite{fedotov2003refinements,reid2009generalised}
  instead of Pinsker's inequality to bound the statistical distance in 
  terms of the Kullback--Leibler divergence,
  which we in turn bound by the mutual information using 
  Lemma~\ref{lemma:KL bound} (with $\epsilon= C\cdot\log t$ for a 
  constant~$C$ large enough so that $I(f(X):X)\leq \epsilon t$ holds):
  \begin{align*}
    &\Exp_j\;\Exp_{x}
    d
    \Big(
      f
      \big(
      X
      \big)
      \;,\;
      f
      \big(
      X
      \vert_{j\leftarrow x}
      \big)
    \Big)
    \\
    &\leq
    \Exp_{j}
    \Exp_{x}
    \left(
      1
      -
      \exp
      \Big(
      -1
      -
      \KLdiv
      \big(
      f
      (
      X
      )
      \;||\;
      f
      (
      X
      \vert_{j\leftarrow x}
      )
      \big)
      \Big)
      \right)
    & \text{(Vajda's inequality)}\\
    &\leq
      1
      -
      \exp\left(-1-
    \Exp_{j}
    \Exp_{x}
      \KLdiv
      \Big(
      f
      \big(
      X
      \big)
      \;||\;
      f
      \big(
      X
      \vert_{j\leftarrow x}
        \big)
        \Big)
        \right)
    & \text{(Jensen's inequality)}\\
    &\leq 1- e^{-1-C\log t}= 1-1/\poly(t) & \text{(Lemma~\ref{lemma:KL 
        bound})}
  \end{align*}
  Now \eqref{eq:vajda} follows from the triangle inequality as 
  follows.
  \begin{align*}
    \Exp_{j}
    \Exp_{x}
    d
    \Big(
      f
      \big(
      X
      \vert_{X_j\neq x}
      \big)
      \;,\;
      f
      \big(
      X
      \vert_{X_j=x}
      \big)
    \Big)
    &\leq
    \Exp_{j}
    \Exp_{x}
    d
    \Big(
      f
      \big(
      X
      \vert_{X_j\neq x}
      \big)
      \;,\;
      f
      \big(
      X
      \big)
    \Big)
    \\
    &\quad+
    \Exp_{j}
    \Exp_{x}
    d
    \Big(
      f
      \big(
      X
      \big)
      \;,\;
      f
      \big(
      X
      \vert_{X_j=x}
      \big)
    \Big)
    \\
    &
    \leq
    \frac{1}{|\Sigma|}+1-\frac{1}{\poly(t)}
    \,.
  \end{align*}
  For this, note that a simple calculation from \eqref{def:statdist} 
  shows that
  \begin{align*}
    &d
    \big(
      f
      (
      X
      \vert_{X_j\neq x}
      )
      ,
      f
      (
      X
      )
    \big)
    \leq
    d
    \big(
      X
      \vert_{X_j\neq x}
      ,
      X
    \big)
    \\
    &\quad\leq
    \Pr(X_j\neq x)
    \cdot
    d
    \big(
      X
      \vert_{X_j\neq x}
      ,
      X
      \vert_{X_j\neq x}
    \big)
    +
    \Pr(X_j= x)
    \cdot
    d
    \big(
      X
      \vert_{X_j\neq x}
      ,
      X
      \vert_{X_j= x}
    \big)\\
    &\quad\leq
    \Pr(X_j\neq x) \cdot 0
    +
    \Pr(X_j = x) \cdot 1
    =
    \Pr(X_j = x)
  \end{align*}
  always holds, and the latter equals $\frac{1}{|\Sigma|}$
  since~$X_j$ is uniformly distributed on~$\Sigma$.
\end{proof}

\section{Extension: \texorpdfstring{$\bm{f}$}{f}-compression}

We end this paper with a small observation: Instead of OR-compressions 
or AND-compressions, we could just as well consider $f$-compressions for a 
Boolean function $f\colon\zo^t\to\zo$.
If the function $f$ is symmetric, that is, if $f(x)$ depends only on 
the Hamming weight of~$x$, then we can represent $f$ as a function 
$f\colon\set{0,\dots,t}\to\zo$.
We make the observation that Drucker's theorem applies to 
$f$-compressions whenever $f$ is a non-constant, symmetric function.
\begin{definition}
Let $f\colon\set{0,\dots,t}\to\zo$ be any function.
Then an $f$-compression of $L$ into $L'$ is a mapping
\begin{align*}
  A\colon\binom{\zo^n}{\leq t} \to \zo^{\epsilon t}\,,
\end{align*}
such that, for all $x\in \binom{\zo^n}{\leq t}$, we have $A(x)\in L'$ 
if and only if $f(|x\cap L|)=1$.
\end{definition}
Examples:
\begin{itemize}
  \item OR-compressions are $f$-compressions with $f(i)=1$ if and only 
    if $i>0$.
  \item AND-compressions are $f$-compressions with $f(i)=1$ if and 
    only if $i=t$.
  \item Majority-compressions are $f$-compressions with $f(i)=1$ if 
    and only if $i>t/2$.
  \item Parity-compressions are $f$-compressions with $f(i)=1$ if and 
    only if $i$ is odd.
\end{itemize}
We can apply Theorem~\ref{thm main} and~\ref{thm main tlogt} 
whenever $f$ is not a constant function.
\begin{lemma}
  Let $f\colon\set{0,\dots,t}\to\zo$ be non-constant.
  Then every $f$-compression for $L$ with size $\epsilon t$ can be 
  transformed into a compression for $L$ or for $\overline{L}$, in the 
  sense of Theorem~\ref{thm main} and with size bound at most 
  $2\epsilon t$.
\end{lemma}
\begin{proof}
  Let $A$ be an $f$-compression from $L$ into $L'$.
  Then $A$ is also a $(1-f)$-compression from $L$ into 
  $\overline{L'}$,
  an $(f(t-i))$-compression from $\overline L$ into $L'$, and
  a $(1-f(t-i))$-compression from $\overline L$ into $\overline{L'}$.
  Since $f$ is not constant, at least one of these four views 
  corresponds to a function $f'$ for which there is an index $i\leq 
  t/2$ so that
  $f'(i)=0$ and $f'(i+1)=1$,
  holds.
  Assume without loss of generality that this holds already for~$f$.
  Then we define $A'\colon\binom{\zo^n}{\leq t-i}\to\zo^{\epsilon t}$ as 
  follows:
  \begin{align*}
    A'(\set{x_{i}, x_{i+1}, \dots, x_t})
    &\doteq
    A(\set{\top_1, \dots,\top_{i-1}, x_{i}, x_{i+1}, \dots, x_t})
    \,,
  \end{align*}
  where $\top_1,\dots,\top_{i-1}$ are arbitrary distinct yes-instances 
  of $L$.
  For the purposes of Theorem~\ref{thm main}, these instances 
  can be written in the non-uniform advice of $A'$.
  If this many yes-instances do not exist, then the language~$L$ is 
  trivial to begin with.
  To ensure that the $x_j$'s are distinct from the $\top_j$'s, we 
  actually store a list of $2t$ yes-instances $\top_j$ and inject 
  only~$i-1$ of those that are different from the $x_j$'s.
  
  $A'$ is just like $A$, except that $i-1$ inputs have already been 
  fixed to yes-instances.
  Then $A'$ is a compressive map that satisfies the following:
  If $|x\cap L|=0$ then $A'(x)\not\in L'$, and if $|x\cap L|=1$ then 
  $A'(x)\in L'$.
  Since the number of inputs has decreased to $t'=t-i\geq t/2$, the 
  new size of the compression is $\epsilon t \leq 2 \epsilon t'$ in 
  terms of $t'$.
\end{proof}

\paragraph{Acknowledgments.}
I would like to thank Andrew Drucker, Martin Grohe, and others for encouraging 
me to pursue the publication of this manuscript, David Xiao for pointing out 
Theorem~\ref{theorem:SD easy} to me, Andrew Drucker, D\'aniel Marx, and 
anonymous referees for comments on an earlier version of this paper, and Dieter 
van Melkebeek for some helpful discussions.

\printbibliography
\end{document}